\documentclass[journal]{IEEEtran}
\usepackage{times}
\usepackage{graphicx}
\usepackage{amssymb, amsmath}
\usepackage{ifthen}
\usepackage{fancyhdr}
\usepackage{bbm}
\usepackage{tikz}
\usepackage{cite}
\usepackage{enumerate}
\usepackage{verbatim}

\newboolean{draft}
\setboolean{draft}{true}

\ifthenelse{\boolean{draft}}
{
}
{

 \setlength{\textwidth}{15cm}
 \setlength{\oddsidemargin}{0.1cm}
 \setlength{\evensidemargin}{0.1cm}
 \setlength{\textheight}{25cm}
}

\newtheorem{thm}{Theorem}[section]
\newtheorem{cor}[thm]{Corollary}
\newtheorem{lem}[thm]{Lemma}

\newtheorem{defn}[thm]{Definition}
\newtheorem{rem}[thm]{Remark}
\newcommand{\be}{\begin{equation}}
\newcommand{\ee}{\end{equation}}
\newcommand{\ben}{\begin{equation*}}
\newcommand{\een}{\end{equation*}}
\newcommand{\ba}{\begin{eqnarray}}
\newcommand{\ea}{\end{eqnarray}}

\newcommand{\Real}{\mathbb R}

\newcommand{\T}{\mathrm{T}}
\newcommand{\E}{\mathbb E}

\newcommand{\iden}{\mathbf{I}}

\newcommand{\supp}{\mathrm{supp}}

\newcommand{\rank}{\mathrm{rank}}
\newcommand{\ind}{\mathbbm 1}

\def\bx{\mathbf{x}}
\def\bA{\mathbf{A}}

\graphicspath{{.}
{./Figures/}
}

\begin{document}
\title{\mbox{Number of Measurements in Sparse Signal Recovery}}

\author{
	Paul Tune, Sibi Raj Bhaskaran, Stephen Hanly \\
	University of Melbourne, VIC-3010, Australia. \\
	{\tt \{l.tune,sibib,hanly\}@unimelb.edu.au}
	\vspace*{-0.2cm}
}

\maketitle

\begin{abstract}
We analyze the asymptotic performance of sparse signal recovery from
noisy measurements. In particular, we generalize some of the 
existing results for the Gaussian case to subgaussian and other
ensembles. An achievable result is presented for the linear sparsity
regime. A converse on 
the number of required measurements in the sub-linear regime is also
presented, which cover
 many of the widely used measurement ensembles. Our converse
idea makes use of a correspondence between compressed sensing ideas
and compound channels in information theory.  
\end{abstract}

%
\section{Introduction}
\label{sec:intro}
Sparse support recovery has been given much attention of late, due to the fact 
that many signals dealt with are sparse in some basis. We will consider the model,
\be
\mathbf{y} = \mathbf{A}\mathbf{x} + \mathbf{z}
\label{eq:model_one}
\ee
where $\mathbf{y} \in \Real^m$, $\mathbf{A} \in \Real^{m\times n}$, $\mathbf{z} 
\in \Real^m$, distributed with $\mathcal{N}(0,\sigma^2 \iden)$. The support of
$\mathbf{x}$ is the index set $\mathcal{I}$, $\supp(\mathbf{x}) =  |\mathcal{I}| = k$.
The signal power, $\|\mathbf{x}\|^2_{\ell_2} = P$. Each column of $\mathbf{A}$ is
normalized to have unit $\ell_2$ norm . 

Our main motivation in this paper is to study a wider class of measurement matrices.
Previous studies have specifically focussed on the Gaussian measurement matrix
\cite{Wainwright07, Akcakaya07}. Two distinct sparsity regimes are often considered in
literature:
\begin{itemize}
\item \textbf{Sublinear}: $\frac{k}{n} \to 0$ as both $k, n \to \infty$, and
\item \textbf{Linear}: $k = \rho n$ for $\rho \in (0,1)$.
\end{itemize}
The following three performance estimates were studied in \cite{Akcakaya07},\cite{Wainwright07}.
\begin{itemize}
\item \textbf{Error metric 1}:
\ben
d_1(\mathbf{x}, \mathbf{\hat x}) 
= \ind \left( \{{\hat x}_i \ne 0\ \forall i \in
\mathcal{I} \} \cap \{\hat x_j = 0\ \forall j \notin \mathcal{I} \} \right)
\een
\item \textbf{Error metric 2}:
\ben
d_2(\mathbf{x}, \mathbf{\hat x}) = \ind \left(\frac{|\{\hat x_i \ne 0\}|}{|
\mathcal{I}|} > 1-\alpha \right)
\een 
\item \textbf{Error metric 3}:
\ben
d_3(\mathbf{x}, \mathbf{\hat x}) = \ind \left(\sum_{k\in \{i|{\hat x}_i \ne 0\} \cap
\mathcal{I}} |x_k|^2 > (1-\delta)P\right)
\een
\end{itemize}
where $\ind(\cdot)$ is the binary valued indicator function which is unity when the
argument is true,  and $\alpha, \delta$ are in $(0,1)$. In Section 
\ref{sec:achievability}, we focus on subgaussian measurement matrices. 
\begin{defn}
A random variable $x$ is {\em subgaussian} if there is a constant $B > 0$ such that
\ben
\Pr(|x| \ge t) \le 2 \exp(-t^2/B^2)
\een
for all $t > 0$. The smallest $B$ is called the {\em subgaussian moment} of $x$.
\label{def:subgaussian}
\end{defn}
An example of a subgaussian measurement matrix is the matrix with i.i.d.~entries of $\pm
1/\sqrt{m}$ distributed according to Bernoulli($\frac{1}{2}$).

 We show that centered 
subgaussian measurement matrices achieve the same asymptotic results as Gaussian 
measurement matrices in the linear sparsity regime, i.e. $m = O(k)$ measurements
suffice for signal recovery. For the linear case, we are taking the \textit{pessimistic}
point of view that good measurement (sensing) schemes should have an exponentially
decaying error probability in the number of measurements, which will also have a 
bearing on the practical constructions. On the other hand, if we take an 
\textit{optimistic} (see \cite{CsiszarKorner81}) viewpoint, that a sub-exponential
decay in error is acceptable, our analysis remains valid for the sub-linear regime also.   

In Section \ref{sec:num:converse}, we present some converse results, which lower bounds
the required number of measurements for asymptotically exact support recovery. Our
converse results give the required scaling of $m$ with respect to $n$ and $k$ in both 
the regimes. Specifically, we invoke a correspondence between compressed sensing schemes
and compound channels in information theory.  Here we
consider general measurement matrices and the underlying assumptions are mild.

\section{Achievability}
\label{sec:achievability}

Our setup for achievability is similar to \cite{Akcakaya07}. In particular, we
extend Theorems~2.1, 2.5 and 2.9 from \cite{Akcakaya07} which provide results for the
number of measurements needed using Gaussian measurement matrices for the error metrics
considered in this paper. For Gaussian measurement matrices, the number of measurements
required for all three error metrics in the linear sparsity regime is $m = O(k)$, where 
the hidden constant value differs for each error metric. For completeness, we state
Theorem 2.1 from \cite{Akcakaya07} here.
\begin{thm}[Achievability for error metric 1]
Let a sequence of sparse vectors, $\{\mathbf{x}_{(n)} \in \Real^n\}_n$ 
($\mathbf{x}_{(n)}$ denotes a dependence on $n$) with $\supp(\mathbf{x}_{(n)}) = k =
\lfloor \rho n \rfloor$. Then asymptotic reliable recovery is possible for
$\{\mathbf{x}_{(n)}\}$ with respect to error metric 1 if $\frac{k \mu^4 (
\mathbf{x}_{(n)})}{\log k} \to \infty$ as $k\to \infty$ and $$ m > c_1 k$$
where $\mu(\mathbf{x}) = \min_{i \in \mathcal{I}} |x_i|$ and $c_1$ is a constant 
depending on $\rho$, $\mu(\mathbf{x})$ and $\sigma$.
\end{thm}
Our result here shows that these results apply to subgaussian measurement matrices. 

On the other hand, in the sublinear sparsity regime, measurements required are now in
the order of $m = O(k\ log(n-k))$ for all three error metrics for Gaussian measurement
matrices. As mentioned earlier, if we take an optimistic viewpoint, then subgaussian
measurement matrices also achieves the same performance as the Gaussian counterpart. The
Lasso scheme was shown to perform optimally in the sublinear regime \cite{Wainwright06}
but the results show that there is a significant gap of the performance of Lasso in
the linear regime.

Let $\mathcal{D} (\mathbf{y})$ be a decoder, which outputs a set
of indices, depending on the problem objective. Our achievability results show the
existence of asymptotically good measurement matrices. Similar to the random coding
arguments in information theory, the average error probability attained by using random
measurement matrices chosen from an ensemble can be made arbitrarily small
asymptotically. However, good matrices are not explicitly identified.

The probability of decoding error for $\mathcal{D}$, averaged over all
measurement matrices $\mathbf{A}$, is defined as
\ben
p_{err}(\mathcal{D}|\mathbf{x}) = \E_{\mathbf{A}}(p_{err}(\mathbf{A}|\mathbf{x}))
= \E_{\mathbf{A}}(\Pr(\mathcal{D}(\mathbf{y}) \ne \mathcal{I})).
\een
We focus on decoders using joint typicality. We define the projection matrix of
$\mathbf{B}$ as $\Pi_\mathbf{B} = \mathbf{B} (\mathbf{B}^\T 
\mathbf{B})^{-1}\mathbf{B}^\T$. The orthogonal projection is defined as 
$\Pi_\mathbf{B}^{\bot} = \iden - \mathbf{B} (\mathbf{B}^\T \mathbf{B})^{-1} 
\mathbf{B}^\T$. 

\begin{defn}[Joint Typicality]\cite{Akcakaya07}
The noisy observation vector $\mathbf{y}$ and a set of indices $\mathcal{J} \subset 
\{1,2,\ldots,n\}$, with $|\mathcal{J}| = k$, are $\delta$-jointly typical if
$\rank(\mathbf{A}_\mathcal{J}) = k$ and 
\ben
\left| \frac{1}{m}\|\Pi^{\bot}_{\mathbf{A}_\mathcal{J}}\mathbf{y}\|_{\ell_2}^2 - 
\frac{m-k}{m}\sigma^2 \right| < \delta
\een
\end{defn}

Denote the events,
\ben
\Omega_{\mathcal{J}} = \{\text{$\mathbf{y}$ and $\mathcal{J}$ are $\delta$-typical}\}
\een
and
\ben
\Omega_0 = \{\rank(\mathbf{A}_{\mathcal{I}}) < k\}.
\een

The decoder has three sources of error:
\begin{itemize}
\item the decoder searches incorrect subspaces, event $\Omega_0$,
\item the true support set $\mathcal{I}$ is not $\delta$-jointly typical, event 
$\Omega_{\mathcal{I}}^c$, and 
\item the decoder recovers another support set $\mathcal{J}$ such that $\mathcal{J} \ne
\mathcal{I}$, event $\Omega_{\mathcal{J}}$. 
\end{itemize}
Hence, the upper bound to the decoder error is given by union bound of the three sources 
of error,
\begin{align}
\label{eq:err_upper_bnd}
p_{err}(\mathcal{D}|\mathbf{x}) 
&\le \Pr(\Omega_0) + \Pr(\Omega_{\mathcal{I}}^c) +
\sum_{\mathcal{J},\mathcal{J}\ne \mathcal{I}, |\mathcal{J}| = k} 
\Pr(\Omega_{\mathcal{J}}).
\end{align}
It suffices to find bounds on each error probability that vanishes asymptotically as 
$n \to \infty$. We show this below.
\subsection{Proof of Achievability}
We first find bounds on the probability that $\Omega_0$ occurs by using the following
result \cite[Theorem 1.1]{Rudelson08tall}.
\begin{lem}
Let $X$ be a subgaussian random variable with zero mean, variance one and subgaussian
moment $B$. Let $\mathbf{X} \in \Real^{m\times k},\ m \ge k$ be the random
matrix whose entries are i.i.d.~copies of $X$. Then there are positive constants $c_1,
c_2$ (depending polynomially on $B$) such that for any $t > 0$
\ben
\Pr(s_k(\mathbf{X}) \le t(\sqrt{m} - \sqrt{k-1})) \le (c_1 t)^{m-k+1} + e^{-c_2 m} .
\een
where $s_k(\mathbf{X})$ denotes the smallest singular value of $\mathbf{X}$.
\label{lem:singular_subgauss}
\end{lem}

In particular, the above lemma suggests that for subgaussian matrices, there is an 
exponentially small positive probability that $s_n =0$. We use this in the following
result.
\begin{thm}
Assume $m > k$. Given an index set $\mathcal{I} \subset \{1,2,\ldots,n\}$ with
$|\mathcal{I}| = k$, 
\ben
\Pr(\rank(\mathbf{A}_{\mathcal{I}}) < k) \le e^{-c_0 m}
\een
for some constant $c_0 > 0$.
\label{thm:bernoulli_singular}
\end{thm}
\begin{proof}
To ensure recovery of $\mathbf{x}$, it is essential that 
$\rank(\mathbf{A}_{\mathcal{I}}) = k$ or equivalently, the smallest singular value,
$s_k(\mathbf{A}_{\mathcal{I}}) \ne 0$. Using Lemma \ref{lem:singular_subgauss}, and
choosing small $t$, we have
\begin{align*}
\Pr(s_k(\mathbf{A}_{\mathcal{I}}) = 0) &= \lim_{t\to 0} 
\Pr(s_k(\mathbf{A}_{\mathcal{I}}) \le t(\sqrt{m} - \sqrt{k-1})) \\
&\le e^{-c_0 m} .
\vspace*{-1.5cm} 
\end{align*}
\end{proof}
\begin{rem}
Reference~\cite{Akcakaya07} uses the fact that if $\mathbf{A}$ has i.i.d.~entries with
 $\mathcal{N}(0,1)$, then $\Pr(\rank(\mathbf{A}_{\mathcal{I}}) < k) = 0$, i.e.,
 $\mathbf{A}_{\mathcal{I}}$ can never be singular. For subgaussian matrices, it is 
possible for such an error to occur. 
For example, with the random sign matrices distributed
according to Bernoulli($\frac{1}{2}$), it is easy to see that
\ben
\Pr(\rank(\mathbf{A}_{\mathcal{I}}) < k) \ge \left( \frac{1}{2} \right)^k.
\een
Hence, Theorem~\ref{thm:bernoulli_singular} says that in the linear regime, 
the error decay for the event $\Omega_0$  is exponential with the number of
 measurements $m>k$. However, a sub-exponential decay to zero can be achieved even
for the sublinear case. The rest of our arguments are valid for both cases.  
\end{rem}
We first modify Lemma 3.3 from \cite{Akcakaya07} by introducing conditions under which
the result is still valid. We then show that the  subgaussian measurement matrices satisfy 
these conditions.
\begin{lem}
\begin{enumerate}
\item Let $\mathcal{I} = \supp(\mathbf{x})$ and assume that $\rank( 
\mathbf{A}_{\mathcal{I}}) = k$. Then for $\delta > 0$,
\begin{align*}
&\Pr\left( \left| \frac{1}{m}\|\Pi^{\bot}_{\mathbf{A}_{\mathcal{I}}} \mathbf{y} \|
^2_{\ell_2} - \frac{m-k}{m} \sigma^2 \right| > \delta \right) \\
&\hspace{1cm} \le 2 \exp\left(-\frac{\delta^2}{4 \sigma^4} \frac{m^2}{m-k + 
\frac{2\delta} {\sigma^2}m} \right).
\end{align*}
This result holds for any measurement matrix $\mathbf{A}$.
\item Let $\mathcal{J}$ be an index set such that $|\mathcal{J}| = k$ and $|\mathcal{I} 
\cap \mathcal{J}| = p < k$, where $\mathcal{I} = \supp(\mathbf{x})$ and assume that 
$\rank(\mathbf{A}_{\mathcal{J}}) = k$. Let 
\ben
V = \frac{\|\Pi^{\bot}_{\mathbf{A}_{\mathcal{J}}} \mathbf{y}\|^2_{\ell_2}}{\sigma^2_y}
-(m-k)
\een
where $\sigma^2_y = \sum_{i\in \mathcal{I} \backslash \mathcal{J}} x_i^2 + \sigma^2$. 
Then $\mathbf{y}$ and $\mathcal{J}$ are $\delta$-joint typical with probability
\begin{align*}
&\Pr\left( \left| \frac{1}{m}\|\Pi^{\bot}_{\mathbf{A}_{\mathcal{J}}} \mathbf{y} \|
^2_{\ell_2} - \frac{m-k}{m} \sigma^2 \right| < \delta \right)\\ 
&\hspace{1cm} \le 2 \exp \left(- \frac{1}{2\gamma_1} \left( (m-k)
\left(1-\frac{\sigma^2}{\sigma^2_y} \right) - \frac{\delta}{\sigma^2_y} m \right)^2
\right)
\end{align*}
if the moment condition 
\be
\log \E\lbrack e^{tV} \rbrack \le -\gamma_1 t - \frac{\gamma_1}{2}\log (1-\gamma_2 t) 
\label{eq:moment_cond}
\ee
is satisfied with constants $\gamma_1,\ \gamma_2 > 0$ for $t < 1/\gamma_2$. 
\end{enumerate}
\label{lem:error_bounds}
\end{lem}
\begin{proof}
The proof for the first item is the same as that of the proof of the first part given in
\cite[Lemma 3.3]{Akcakaya07}. We have
\ben
\Pi^{\bot}_{\mathbf{A}_{\mathcal{I}}} \mathbf{y} = \Pi^{\bot}_{\mathbf{A}_{\mathcal{I}}}
\mathbf{z},
\een
and
\ben
\Pi^{\bot}_{\mathbf{A}_{\mathcal{J}}} \mathbf{y} = \Pi^{\bot}_{\mathbf{A}_{\mathcal{J}}}
\left(\sum_{i\in \mathcal{I}\backslash\mathcal{J}} x_i \mathbf{a}_i + \mathbf{z} 
\right).
\een
It can be seen that, by the property of symmetric projection matrices, 
$\Pi^{\bot \T}_{\mathbf{A}_{\mathcal{I}}} \Pi^{\bot}_{\mathbf{A}_{\mathcal{I}}} = 
\Pi^{\bot}_{\mathbf{A}_{\mathcal{I}}}$. Furthermore, $\mathbf{z}$ is independent of the
entries of $\Pi^{\bot}_{\mathbf{A}_{\mathcal{I}}}$. Hence 
by \cite[Chapter 18]{Johnson95vol1}, 
\ben
\frac{\|\Pi^{\bot}_{\mathbf{A}_{\mathcal{I}}} \mathbf{z}\|^2_{\ell_2}}{\sigma^2} = 
\left(\frac{\mathbf{z}}{\sigma}\right)^\T \Pi^{\bot}_{\mathbf{A}_{\mathcal{I}}} 
\left(\frac{\mathbf{z}}{\sigma}\right) \sim \chi^2_{m-k}
\een
By using concentration inequalities of chi-squared random variables around their degrees 
of freedom ($m-k$ here) as in \cite[Lemma 3.3]{Akcakaya07}, the same result is obtained. 

For the second part of the lemma, we have
\begin{align*}
&\Pr\left(\left|\frac{1}{m}\|\Pi^{\bot}_{\mathbf{A}_\mathcal{J}}\mathbf{y}\|_{\ell_2}^2 
- \frac{m-k}{m}\sigma^2 \right| < \delta \right)\\
&= \Pr\left(\frac{1}{m}\|\Pi^{\bot}_{\mathbf{A}_\mathcal{J}}\mathbf{y}\|_{\ell_2}^2 - 
\frac{m-k}{m}\sigma^2 < \delta \right)\\
&\hspace{1cm} + \Pr\left(\frac{1}{m}\| \Pi^{\bot}_{\mathbf{A}_\mathcal{J}}\mathbf{y}\|_
{\ell_2}^2 - \frac{m-k}{m}\sigma^2 > -\delta \right)\\
&= \Pr\left(V < -(m-k)\left(1 - \frac{\sigma^2}{\sigma^2_y}\right) 
+ \frac{\delta} {\sigma_y^2}m \right)\\
&\hspace{1cm} + \Pr\left(V > -(m-k)\left(1
-\frac{\sigma^2} {\sigma^2_y} \right) - \frac{\delta}{\sigma_y^2}m \right).
\end{align*}
Using Chernoff's bound and the moment condition, it can be shown that for any $\lambda >
0$ (see Appendix), 
\be
\Pr(V \ge \gamma_2 \lambda +  \sqrt{2 \gamma_1\lambda}) \le 
\Pr(V \ge  \sqrt{2 \gamma_1\lambda}) \le e^{-\lambda}
\label{eq:ge}
\ee
and
\be
\Pr(V \le -\sqrt{2\gamma_1 \lambda}) \le e^{-\lambda}.
\label{eq:le}
\ee
We bound the first probability by choosing in equation \eqref{eq:ge},
\ben
\lambda_1 = \frac{1}{2\gamma_1} \left( (m-k) \left(1-\frac{\sigma^2}{\sigma^2_y} 
\right) - \frac{\delta}{\sigma^2_y} m \right)^2
\een
and for the second,
\ben
\lambda_2 = \frac{1}{2\gamma_1} \left( (m-k) \left(1-\frac{\sigma^2}{\sigma^2_y} 
\right) + \frac{\delta}{\sigma^2_y} m \right)^2
\een
in equation (\ref{eq:ge}), we have
\ben
\Pr\left(\left|\frac{1}{m}\|\Pi^{\bot}_{\mathbf{A}_\mathcal{J}}\mathbf{y}\|_{\ell_2}^2 
- \frac{m-k}{m}\sigma^2 \right| < \delta \right) \le 2 \exp(-\lambda_1).
\een
since $\lambda_1 \le \lambda_2$.
\end{proof}

\begin{thm}
Subgaussian measurement matrices satisfy Lemma \ref{lem:error_bounds} with $\gamma_1 = 
m-k$ and $\gamma_2 = 2$.
\label{thm:subgauss_cond}
\end{thm}
\begin{proof}
We only need to show how subgaussian measurement matrices satisfy Lemma 
\ref{lem:error_bounds}(2). We first note that subgaussian r.v.s have a closure property
under addition. Hence, the vector
\ben
\mathbf{y} = \sum_{i\in \mathcal{I}\backslash\mathcal{J}} x_i \mathbf{a}_i + \mathbf{z}
\een
is still subgaussian since for some constant $\alpha$ \cite{Mikosch91},
\begin{align*}
\E\lbrack e^{t\mathbf{y}} \rbrack 
&\le \exp\left( \frac{t^2}{2}(\sum_{i\in \mathcal{I} \backslash \mathcal{J}} x_i^2 
\alpha^2 + \sigma^2) \mathbf{1} \right)
\le \exp\left( \frac{t^2}{2}(\alpha' \sigma_y^2) \mathbf{1} \right)
\end{align*}
where $\mathbf{1}$ is the column vector of 1s, $\alpha' > 0$ is a constant and
\ben
\sigma_y^2 = \sum_{i\in \mathcal{I} \backslash \mathcal{J}} x_i^2 + \sigma^2.
\een
Note that the vector is independent of the entries of $\Pi^{\bot}_ {\mathbf{A} 
_\mathcal{J}}$. 

Since $\Pi^{\bot}_{\mathbf{A}_\mathcal{J}}$ is a symmetric and idempotent,
we rewrite
\ben
\frac{\|\Pi^{\bot}_{\mathbf{A}_{\mathcal{J}}}\mathbf{y} \|^2_{\ell_2}}{\sigma_y^2} =
\left( \frac{\mathbf{y}}{\sigma_y} \right)^\T \Pi^{\bot}_{\mathbf{A}_\mathcal{J}} \left(
\frac{\mathbf{y}}{\sigma_y} \right).
\een

To bound the moment, we require an estimate using \cite[Lemma 1.2]{Mikosch91}, for $0
\le t < 1/(2\alpha')$,
\ben
\E\left[ \exp (tV) \right] \le e^{-t (m-k)} \cdot (1-2t)^{-(m-k)/2}.
\een
Note that the upper bound is the moment generating function of distribution
$\chi^2_{m-k}$. 

The function $\log \E\lbrack \exp (tV)\rbrack$ is monotonically decreasing in $t < 0$
and at $t=0$, we have $\log \E\lbrack \exp (tV) \rbrack \le 0$. On the other hand, the
function $(m-k)t^2$ is monotonically increasing for $t < 0$. As such, we have $\log
\E\lbrack \exp (tV) \rbrack \le (m-k)t^2$ for $t < 0$. Hence, it can be easily seen that
$\gamma_1 = m-k$ and $\gamma_2 = 2$. 
\end{proof}
With $\gamma_1 = m-k$, $\gamma_2 = 2$, $\hat{\sigma}= 1 - \frac{\sigma^2}{\sigma_y^2}$ and 
$\delta^{\prime} = \delta m/(m-k)$,
\begin{align*}
\Pr(\Omega_{\mathcal{J}}) &\le 2 \exp \left(-\frac{1}{2(m-k)} \left( (m-k)
\hat{\sigma} - \frac{\delta}{\sigma^2_y} m \right)^2
\right) \\
&\le 2 \exp \left(-\frac{m-k}{4} \left( \frac{\sigma_y^2 - \sigma^2 - \delta'}
{\sigma_y^2} \right)^2 \right) \\
&= 2\exp \left(-\frac{m-k}{4} \left(\frac{\sum_{k \in \mathcal{I}\backslash \mathcal{J}}
 x_k^2 - \delta'}{\sum_{k \in \mathcal{I}\backslash \mathcal{J}} x_k^2 + \sigma^2}
 \right)^2 \right).
\end{align*}
Assuming $\rank(\mathbf{A}_{\mathcal{J}}) = k$, the number of subsets $\mathcal{J}$ that
overlaps $\mathcal{I}$ in $p$ indices is upper-bounded by
$\binom{k}{p}\binom{n-k}{k-p}$,
implying that by (\ref{eq:err_upper_bnd}) and  Theorem~\ref{lem:error_bounds},
\begin{align*}
\nonumber
p_{err}(\mathcal{D}|\mathbf{x}) 
& \le \exp(-c_0 m) + 2 \exp\left(-\frac{\delta^2}{4 \sigma^4} \frac{m^2}{m-k +
\frac{2\delta}{\sigma^2}m} \right)\\
\nonumber
&\hspace*{-1.3cm} + 2 \sum_{p=1}^k \binom{k}{p}\binom{n-k}{k-p} \exp (-\frac{m-k}{4}
\left(\frac{\underset{k \in \mathcal{I}\backslash \mathcal{J}}{\sum} x_k^2 - \delta'}{\underset{k \in
\mathcal{I}\backslash \mathcal{J}}{\sum} x_k^2 + \sigma^2} \right)^2 ).
\end{align*}

We sketch an outline of the rest of the proof here. Only $\Pr(\Omega_{\mathcal{J}})$
changes depending on the error metric. Let $|\mathcal{I} \cap \mathcal{J}| = p$ for some
particular set $\mathcal{J}$. For error metric 1, we note that $\sum_{k \in
\mathcal{I}\backslash \mathcal{J}} x_k^2 \ge (k-p)\mu^2 (\mathbf{x})$. For error
metric 2, since we only need $\Pr(\Omega_{\mathcal{J}}) \to 0$ for $p \le (1-\alpha)k$
for $\alpha \in (0,1)$, then we have $\sum_{k \in \mathcal{I}\backslash \mathcal{J}}
x_k^2 \ge \alpha k \mu^2 (\mathbf{x})$ for error to occur. Finally, for error
metric 3, we have $\sum_{k \in \mathcal{I}\backslash \mathcal{J}} x_k^2 \ge
\gamma P$ for error to occur. The rest of the arguments on bounding the error
probability follows that of the analysis on Gaussian measurement ensembles in
\cite{Akcakaya07}, both in the linear and sublinear sparsity regimes.

%
\section{Converse on the  Number of Measurements \label{sec:num:converse}}
Our starting point is again the signal recovery  model in \eqref{eq:model_one}.
 For simplicity, assume that
$\bx$ has $k$ non-zero entries. Further more, the entries of $\bA$
are taken from some alphabet $\mathcal{A}$, and normalized, i.e., for
each column $a_i$,
\begin{align}
\frac{1}{m} \|a_i\|_{\ell_2}^2 = 1\,\,, \,\, a_{ij} \in \mathcal{A}
\label{eq:alph:constraint}
\end{align}
Note that the measurement matrix $\bA$ is specified in advance
without the knowledge of the instantaneous realization of $\bx$. So
$\bA$ depends only on the global properties of $\bx$ and the noise
statistics. For simplicity (also for practical reasons), we make the
mild assumption  that there is no prior knowledge about the input
values favoring any particular locations. This implies that the
support of $\bx$ is uniformly chosen from the $\binom nk$ possible
choices.
%

Our discussion in this section is for the error metric 1, but can be
tailored for other purposes too. Recall that for the first metric,
we are interested in recovering the support of $\bx$ based on $m$
measurements from \eqref{eq:model_one}. The error probability in
recovering the support lower-bounds that of exact signal recovery.
This can be easily seen by imagining a genie which tells the
receiver  about the non-zero components in the order  of their
appearance.

We need some notation to proceed. Let us define the following: \\
\begin{tabular}{lcp{7.0cm}}
$\bar{\alpha}$ &-& the vector of non-zero values of $\bx$, in descending order
of magnitude, the $i^{th}$ entry being $\alpha_i$. \\
$\beta$ &-& non-zero values of $\bx$ in the order of  appearance. \\
$I_o$ &-& set of indices of $\bx$ with zero magnitude. \\
$\rho(\alpha_i)\!\!\!$ &-& index in $\bx$ corresponding to the $i^{th}$ entry of $\alpha$. \\
\end{tabular}
\begin{tabular}{lp{6.2cm}}
$\bar{R}(k,\bar{\alpha}, \sigma^2)$  - &  capacity  region of a  $k$-user
single antenna  Gaussian MAC with channel gains $\bar{\alpha}$, and input constraints
as in \eqref{eq:alph:constraint}.
\end{tabular}

 Let $\hat{\bx}$ be the recovered vector using some decoding method. In this section, we assume that
 $m$ is large enough, with respect to $k$ and in relation to $n$,
 to ensure that the probability of decoding error tends
 to zero as $m, n$ and $k$ tend to infinity. The error event can be
 written in terms of a random variable $\Phi$, which is defined as,
\begin{align}
\Phi = (\left(\prod_{i:\bx_i=0} \ind_{\{\hat{\bx}_i=0\}} \right). \left(\prod_{i:\bx_i \neq 0} \ind_{\{\hat{\bx}_i \neq 0\}}\right)).
\end{align}

Given the $k$ non-zero symbols $\beta$,
 $\Phi$ is induced by a uniform distribution
on the $\binom nk$ possible supporting indices of the vector $\bx$.
In many practical cases, $\beta$ is drawn from some distribution.
Our results can be extended to handle this, but presently we stick
to fixed $\beta$, and we assume all the components of $\beta$ are
distinct. The later assumption is just for saving some notation, and
has no bearing on the technical details.
 The average error probability now becomes,
\begin{align}
P_{error} = \Pr(\Phi = 0).
\end{align}
The following lemma yields a lower bound on $m$, the number of
measurements required for asymptotically \textit{exact} support
recovery.
\begin{lem} \label{lemma:converse1}For a given $\beta$ with $k$ non-zero elements,
if $P_{error}$ goes to zero with $m$, then
\begin{align}
m \geq \frac{k \log (n/k)}{R_{CMAC}(k,\bar{\alpha},\sigma^2)}
\end{align}
where
\begin{equation}
 R_{CMAC}(k,\bar{\alpha},\sigma^2) =
    \min_{\alpha^*} \max_{R \in \mathbb{R}^k} \|R\|_{\ell_1} . \ind_{\{R \in \bar{R}(k,\alpha^*,\sigma^2)\}}
    \label{eq:CMAC_sumrate}
\end{equation}
and $\alpha^*$ is any permutation of the channel coefficients $\bar{\alpha}$.
\end{lem}
The proof of this lemma proceeds in number of stages. In the next
few paragraphs, we will explain the essential ideas behind it.  The
arguments that we present shed light on some of the underlying
bottlenecks in the detection problem.

%
 To obtain a bound as above, we map the support recovery (SR) problem
to a communication problem and then establish the connection between
the number of measurements $m$ and the required number of channel
uses in the communication model, or alternatively to the maximal
rate at which error-free transmissions are possible.

In principle, the communication setup that we describe can simulate
any strategy for the support recovery problem. 
We briefly describe how any SR problem comes under our communication 
setup, see Figure below.


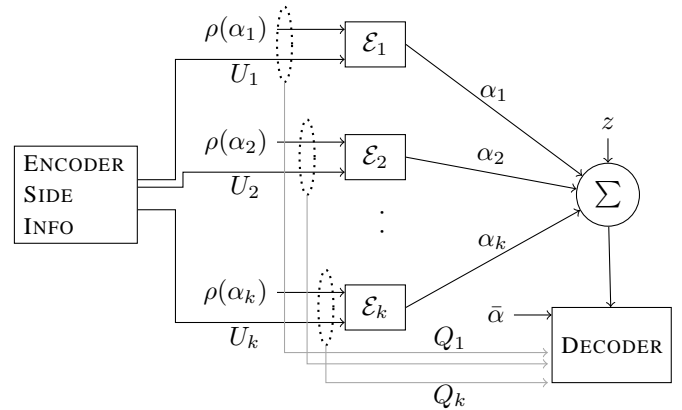
\begin{center}
\tikzstyle{encoder} = [right,rectangle, draw, minimum width=0.8cm, minimum height=0.6cm]
\begin{tikzpicture}
\node(SI) at (-1.25,0) [left,text width=1.4cm, rectangle, draw] {\small \textsc{Encoder Side Info}};
\node (e1) at (1.5,2.0) [encoder] {$\mathcal{E}_1$};
\node (e2) at (1.5,0.5) [encoder] {$\mathcal{E}_2$};
\node  at (2,-0.25) {$\cdot$};
\node  at (2,-0.5) {$\cdot$};
\node (ek) at (1.5,-1.5) [encoder] {$\mathcal{E}_k$};
\draw (SI.east) ++(0,0.2) --++(0.5,0) --++(0,1.6)[->] --++(2.25,0);
\draw (SI.east) ++(0,0.1) --++(0.6,0) --++(0,0.2)[->] --++(2.15,0);
\draw (SI.east) ++(0,-0.2) --++(0.5,0) --++(0,-1.5)[->] --++(2.25,0);
\draw (0.6,2.2) node [left] {$\rho(\alpha_1)$} [->] --++(0.9,0) ;
\draw (0.5,1.6) node [left] {$U_1$} ++(0,-1.5) node[left]{$U_2$} ++(0,-2.0)
				node[left]{$U_k$};
\draw (0.6,0.7) node [left] {$\rho(\alpha_2)$} [->] --++(0.9,0) ;
\draw (0.6,-1.3) node [left] {$\rho(\alpha_k)$} [->] --++(0.9,0) ;
\node (sum) at (5,0) [circle,draw] {$\sum$} 
	edge[<-] node[above]{$\alpha_1$} (e1.east)
	edge[<-] node[above]{$\alpha_2$} (e2.east)
	edge[<-] node[above]{$\alpha_k$} (ek.east) ;
\node (dec) at (4.25,-2.0) [right, rectangle, draw, minimum height=1.0cm] {\small \textsc{Decoder}} edge[<-] (sum);
\draw  (3.75,-1.6) node [left]{$\bar{\alpha}$}[->] --++(0.5,0);
\draw[thick, dotted] (0.7,2.0) ellipse(0.1cm and 0.5cm);
\draw[->,gray!70] (0.7,2.0) ++(0,-0.5) --++(0,-3.6)  --++(3.5,0) node[xshift=-1.3cm,yshift=0.18cm,black]{$Q_1$};
\draw[thick, dotted] (1.0,0.5) ellipse(0.1cm and 0.5cm);
\draw[->,gray!70] (1.0,0.50) ++(0,-0.5) --++(0,-2.25) --++(3.2,0);
\draw[thick, dotted] (1.25,-1.5) ellipse(0.1cm and 0.5cm);
\draw[->,gray!70] (1.25,-1.50) ++(0,-0.5) --++(0,-0.5) --++(2.95,0) node[xshift=-1.3cm,yshift=-0.18cm,black]{$Q_k$} ;
\draw[<-] (sum) --++(0,0.75) node[above]{$z$}; 
\end{tikzpicture}
\end{center}


Recall the notations introduced in paragraph~3 of this section. Consider $k$ encoders
trying to communicate information to a decoder. Each encoder corresponds to a non-zero
value of the input vector $\bx$ in the support recovery problem. Perform a random
permutation of the set $I_o$ and partition it into $k$ subsets $\{U_1,\ldots, U_k\}$,
provide this to each encoder as side-information. The decoder is given the index set $Q_i$ 
of each encoder's inputs, as well as the channel coefficient $\alpha_i$ from that encoder.
Clearly this system can emulate the SR problem. We now take an alternate view to the keep
the discussion as simple as possible. We describe a setup where the sparse vector $\bx$
for the SR problem, and the messages for the correponding communication problem are 
generated together. There is no loss of generality in coupling the two systems like this. 

To this end, randomly permute the indices of $\bx$ and partition them into $k$ sets $S_1,
S_2,\ldots, S_k$. To partially emulate the SR problem, the support of $\bx$ is chosen by
selecting one element from each of these sets, which correspond to the indices of the
support of $\bx$. This selection will correspond to message selection in a $k$-user
communication channel, in which $S_k$ is the message set of user $k$. A simple method of
communication is for user $k$ to encode the chosen message by sending the corresponding
column of $\bA$ directly (rather like a CDMA scheme, with no additional coding) and the
decoder then receives 
$$y = \sum_{i=1}^k \beta_{\pi(i)} a_i + z$$
where $a_i$ is the column corresponding to the message chosen by user $i$, and $\pi$ is a
random permutation of $\{1, 2, \ldots, k\}$ that assigns a component of $\bx$ to user $k$. 
The decoder is given the vector $(\beta_{\pi(i)})_{i=1}^k$ as side information. This
coherent $k$-user faded AWGN communication channel is a partial emulation of the CS
decoding problem in \eqref{eq:model_one}, except that here the decoder has more 
information: it knows that each $S_i$ contains exactly one index from the support of the
vector $\bx$, and it knows the corresponding value of $\bx$ in that component, namely
$\beta_{\pi(i)}$. Note that user $i$ is conveying $\log(|S_i|)$ bits to the decoder, and
the total number of bits being conveyed is $\sum_{i=1}^k \log(|S_i|)$ bits. The decoder in
this communication set-up must do at least as well as the CS decoder in the original
problem, so these bits are being conveyed reliably.

The above simple CDMA communication scheme is valid for the $k$-user, faded AWGN channel
in which, in general, the user is allowed to encode his messages using symbols each taken
from the same alphabet as the symbols in $\bA$, and each codeword satisfies the power
constraint \eqref{eq:alph:constraint}. Since the permutation $\pi$ is selected randomly,
this is a compound MAC, and the rate region can in principle be calculated. In a compound
MAC, the transmitter knows only a set of possible MACs from which one realization will be
picked \cite{CsiszarKorner81}. We do not go into the details of the coding theorems,
rather we merely use the results on the achievable maximal sum-rate. Compound MAC capacity
region is contained in the intersection of MAC capacity regions of the individual
components; in our case, the sum-rate is at best that in \eqref{eq:CMAC_sumrate}. The
lemma is proved by noting that the communication scheme requires successful communication
of $k \log(n/k)/m$ bits per channel-use, when we choose each set $S_i$ to have $n/k$
indices. This rate must be upper bounded by the sum-rate of the compound MAC.

\begin{cor}
By using a Gaussian measurement ensemble,
\begin{align}
m \geq \max \left\{ \frac{2 \log \frac nk}{\log(1 +
\alpha_k^2/\sigma^2)},
    \frac{2 k \log \frac nk}{\log(1+ \|\bar{\alpha}\|_{\ell_2}^2/\sigma^2)}
 \right\}.
\end{align}
and when $\alpha_k/\sigma << 1.0$,
\begin{align}
m \geq \frac {\sigma^2 \log \frac nk}{\alpha_k^2}
\label{eq:bound:ww}
\end{align}
\end{cor}
The corollary follows from Lemma~\ref{lemma:converse1}  by noting
that the maximal sum-rate in the compound MAC  setting is less than
$k \log(1+\frac{\alpha_k^2}{\sigma^2})$, since this is the sum of
the single user constraints. The expression in \eqref{eq:bound:ww}
is identical to that obtained in \cite{Wainwright07}, which  can be
further
 tightened by an alternative approach. Consider the above compound
 MAC, when we take $S_1$ to have size $n-k+1$ and $|S_i|=1 , \forall i>1$.
 In this case, user $1$ is
 conveying $\log(n-k+1)$ bits to the decoder, and the other users
 are conveying zero bits, since the decoder knows apriori that these
 users have only one index (corresponding to telling the CS decoder
 $k-1$ elements of the support set as side information). The single
 user rate constraint then tells us that
\begin{align}
m \geq \frac {\log (n-k+1)}{\log(1+\alpha_k^2/\sigma^2)}.
\end{align}
\begin{cor}
If the measurement matrix is chosen by Bernoulli($\frac 12$) on
$\{+1,-1\}$,
\begin{align}
m \geq \frac{2 k \log_2 \frac nk}{\log_2 \pi e k/2} \label{eq:changwolf}
\end{align}
\end{cor}
With  $\{+1,-1\}$ as the input alphabet, we can see that this
channel has sum-rate strictly less than that available in a $k$ user
binary-input adder channel \cite{ChangWolf81}.  The achievable
 sum-rate there is half that of the denominator in \eqref{eq:changwolf}.
 This bound can be made tighter by considering an adder channel with noise, but we do not
 pursue it here.

Bounding the number of measurement as above also allows us to  get insights 
about the speed at which exponential decay of recovery-error happens, this is
given in the following lemma.
\begin{lem}
The error probability in support recovery obeys,
\begin{align}
P_{error} \geq \exp(-E_0(\alpha_k,\sigma^2) m) ,
\end{align}
where $E_0(\alpha,\sigma^2)$ is the \textit{cut-off rate} of a standard scalar
AWGN channel with power constraint $\alpha^2/\sigma^2$.
\end{lem}
Notice that in the compound MAC we consider,  the error probability in the scalar
channel with gain $\alpha_k$  lowerbounds the total error probability. The best exponent
of error-decay for this channel is given by the above $E_0(\cdot)$, which is also the
maximal error exponent, happening at zero rate. We can extend this result to include the
sphere-packing and straightline bounds,  this is part of some ongoing work.


\section{Related Work}

A direct comparison can be made between our work and that of \cite{Akcakaya07}. In that
paper, it was shown that Gaussian measurement matrices are asymptotically optimal for
joint typical decoders with O($k$) measurements, with fixed SNR, for each error metric
defined here. We extend this result to show that these sufficient conditions also hold
for centered subgaussian measurement matrices in the linear sparsity regime. Necessary
conditions are also established in \cite{Akcakaya07} using arguments based on MACs,
however, their bounds are not as refined as ours.

In \cite{Wainwright07}, necessary and sufficient conditions are given for error metric 1.
Sufficient conditions were established using an ML decoder while the necessary conditions
exploited a corollary of Fano's inequality. By comparing results in \cite{Wainwright07}
and \cite{Wainwright06}, it was shown that, in the sublinear sparsity regime, Lasso is
essentially information theoretically optimal. However, in the linear regime, there has
been no practical algorithm that has achieved the $\Omega(k \log (n-k))$ bound established
in our paper and Fletcher et al.\cite[Theorem 1]{FletcherGoyal08}.

Results from Fletcher et al.\cite{FletcherGoyal08} is the closest to ours in terms of the
scaling bounds they achieved. After submitting a first version here, we noticed that 
\cite{FletcherGoyal08} describes some good bounds for the Gaussian case, along with a
detailed comparison with existing bounds. Our converse bound generalizes their result, and
we believe it is comparable for specific instances. A detailed study along this direction
will be included in the final manuscript.

Partial support recovery was also addressed in \cite{Reeves08} and necessary conditions
are given. There a general bound was derived for deterministic and stochastic signals. A
bound strictly focussed on Fourier measurement matrices is found in
\cite{GastparBresler00}, which uses Fano's inequality to establish the bound. In terms of
the necessary condition in \cite[Theorem 3.2]{Reeves08}, Theorem \ref{lemma:converse1} is
tighter and is also general as it applies to a variety of measurement ensembles. Theorem 
\ref{lemma:converse1} is general enough to apply to structured codewords, such as Fourier
measurement matrices, although the codewords now have a dependence. However, one needs to
compute the capacity region of the a compound MAC channel using these structure codewords.

\section{Conclusion}

We have analyzed schemes for sparse signal recovery using subgaussian measurement 
matrices. Our achievability scheme used an impractical decoder. Future work intends
to tackle the performance of subgaussian matrices and practical decoders.

\section*{Appendix}

We sketch the proof of the concentration result based on modification of arguments by
Birg\'e and Massart in \cite{Birge98}. Let $\epsilon = \gamma_2 \lambda + \sqrt{2 
\gamma_1 \lambda}$. We first prove (\ref{eq:ge}) bounding $V$ using Chernoff's bound,
\ben
\Pr(V \ge \epsilon) \le \exp\left( \inf_{t > 0} \left( -t\epsilon + \log \E\lbrack
e^{tV}\rbrack\right) \right)
\een
Since $V$ satisfies the moment condition,
\ben
\log \E\lbrack e^{tV} \rbrack \le -\gamma_1 t - \frac{\gamma_1}{2}\log (1-\gamma_2 t)  
\le \frac{\gamma_1 t^2}{2(1-\gamma_2 t)} 
\een
we have
\ben
\Pr(V \ge \epsilon) \le \exp \left( -g(\epsilon) \right)
\een
where
\ben
g(\epsilon) = \sup_{t > 0} \left( t\epsilon - \frac{\gamma_1 t^2}{2(1- \gamma_2 t)}
\right) .
\een
It can be shown that the supremum is achieved for $t = \gamma_2^{-1} \lbrack 1- 
\sqrt{\gamma_1} (2\epsilon \gamma_2 + \gamma_1)^{-1/2}\rbrack$ and that
\ben
g(\epsilon) \ge \frac{\epsilon^2}{2\gamma_2 \epsilon + 2 \gamma_1}.
\een
and that $g(\epsilon) = \lambda$. To prove (\ref{eq:le}), we note that $\log \E\lbrack 
e^{tV} \rbrack \le \gamma_1 t^2$ for $-1/\gamma_2 < t < 0$. The result then follows.

\small
\bibliographystyle{abbrv}
\bibliography{isit09}
\nocite{Cover06, CsiszarKorner81}
\end{document}